\documentclass[a4paper,11pt]{article}
\usepackage{modules/convOpt}
\definecolor{grayfill}{RGB}{195,195,195}
\definecolor{darkgrayfill}{RGB}{150,150,150}

\definecolor{hr}{RGB}{85,85,85}
\definecolor{thr}{RGB}{200,200,200}
\definecolor{edHr}{RGB}{190,20,50}
\definecolor{edThr}{RGB}{230,110,130}

\definecolor{fill1}{RGB}{250,156,94}		
\definecolor{fill2}{RGB}{247,253,86}		
\definecolor{fill3}{RGB}{135,246,121}		
\definecolor{fill4}{RGB}{161,248,239}		
\definecolor{fill5}{RGB}{198,142,245}		
\definecolor{fill6}{RGB}{185,168,139}		

\definecolor{highlight}{RGB}{7,26,208}

\def\zal#1{\filldraw[draw=fill#1 , fill=fill#1 , opacity=0.55]}
\def\zall#1{\fill[fill=fill#1 , fill opacity=0.55]}

\begin{document}

\title{A complete list of all convex polyhedra made by gluing regular pentagons%
\thanks{E. A. was supported in part by F.R.S.-FNRS, and by the SNF grant P2TIP2-168563 of the Early PostDoc Mobility program. E.~A. and B.~Z. are partially supported by the Foundation for the Advancement of Theoretical Physics and Mathematics ``BASIS'' and by ``Native towns'', a social investment program of PJSC ``Gazprom Neft''. S. L. is directeur de recherches du F.R.S.-FNRS.}}

\author{Elena Arseneva\thanks{
		St. Petersburg State University (SPbU).
		Emails: e.arseneva@spbu.ru, boris.a.zolotov@yandex.com.
	} \and
	Stefan Langerman\thanks{
		Universit\'e libre de Bruxelles (ULB). Email: stefan.langerman@ulb.ac.be.
	} \and
	Boris Zolotov$^\dagger$}

\maketitle

\begin{abstract}
We give a complete description of all convex polyhedra whose surface can be constructed from several congruent regular pentagons by folding and gluing them edge to edge. Our method of determining the graph structure of the polyhedra from a gluing is of independent interest and can be used in other similar settings.
\end{abstract}

\section{Introduction}
Given a collection of 2D polygons, a \emph{gluing} describes a closed surface by specifying how to glue (a part of) each edge of these polygons onto (a part of) another edge. Alexandrov's uniqueness theorem~\cite{alex} states that any valid gluing that is homeomorphic to a sphere and that does not yield a total facial angle greater than $2\pi$ at any point, corresponds to the surface of a unique convex 3D polyhedron (doubly covered convex polygons are also regarded as polyhedra). Note that the original polygonal pieces might need to be folded to obtain this 3D surface.

Unfortunately, the proof of Alexandrov's theorem is highly non-constructive. The only known approximation algorithm to find the vertices of this polyhedron~\cite{kpd09-approx} has (pseudopolynomial) running time really large in $n$, where $n$ is the total complexity of the gluing. In particular, its running time depends on $n$ as $\tilde{O}(n^{578.5})$, and it also depends on the aspect ratio of the polyhedral metric, the Gaussian curvature at its vertices, and the desired precision of the solution. There is no known exact algorithm for reconstructing the 3D polyhedron, and in fact the coordinates of the vertices of the polyhedron might not even be expressible as a closed formula~\cite{bannister2014galois}.

Enumerating all possible valid gluings is also not an easy task, as the number of gluings can be exponential even for a single polygon~\cite{DDLO02}. However one valid gluing can be found in polynomial time using dynamic programming~\cite{DO07,lo96-dynprog}. Complete enumerations of gluings and the resulting polyhedra are only known for very specific cases such as the Latin cross~\cite{ddlop99} and a single regular convex polygon~\cite{DO07}.

The special case when the polygons to be glued together are all identical regular $k$-gons, and the gluing is \emph{edge-to-edge} was recently studied by the first two authors of this paper~\cite{kl17-hex}. For $k>6$,  the only two possibilities are two $k$-gons glued into a doubly-covered $k$-gon, or one $k$-gon folded in half (if $k$ is even). When $k=6$, the number of hexagons that can be glued into a convex polyhedron is unbounded. However, for non-flat polyhedra of this type there are at most ten possible graph structures. For six structures out of these ten, the gluings realizing them have been found. For doubly-covered 2D polygons, all the possible polygons and the gluings forming them have been characterized.

In this paper we continue this study by thoroughly considering the case of $k=5$, i.e., gluing regular pentagons edge to edge. This setting differs substantially from the case of hexagons, since it is not possible to produce a flat vertex by gluing regular pentagons. Therefore both the number of possible graph structures and the number of possible gluings is finite and little enough to study each one of them individually. 

We start by enumerating  all edge-to-edge gluings of regular pentagons satisfying the conditions of the Alexandrov's Theorem (Section~\ref{sec:enum}). After that we solve the problem of establishing the graph structure of the convex polyhedra corresponding to each such gluing $G$. Using the existing methods (implementation~\cite{sech} of the Bobenko-Izmestiev algorithm~\cite{boben}), we obtain an approximate polyhedron $P$ for gluing $G$. With the help of a computer program, we generate a certificate that the edges of these approximate polyhedra are present in the sought polyhedra. In particular, we upper bound the discrepancy in vertex coordinates between the unique convex polyhedron corresponding to $G$ a given approximate polyhedron (Theorem~\ref{thm:precision}), which implies a sufficient condition for the polyhedron to have a certain edge (Theorem~\ref{thm:code-whattocheck}). Our computer program checks this condition automatically. For non-simplicial approximate polyhedra $P$, to prove that there are no additional edges present in the sought polyhedra, we resort to ad-hoc geometric methods, using symmetry arguments and reconstructing the process of gluing the polyhedron (Section~\ref{section:geom}).

While the main outcome of this work is the full list of the convex polyhedra that are obtained by gluing regular pentagons edge to edge (Section~\ref{section:complete}), the methods for obtaining it are of independent interest and may be applied to other problems of the same flavour.

\section{Preliminaries and definitions}

In this section we review definitions and previous results that are necessary for the rest of this paper. We start with some basic notions.

By a polyhedron we mean a three-dimensional polytope, and, unless stated otherwise, all the polyhedra we are considering are convex. Doubly-covered convex polygon is also regarded as a convex polyhedron. A polyhedron is called \emph{simplicial} if all its faces are triangles.

Consider an edge $e$ of a polyhedron; and let  $f_1$ and $f_2$ be the two faces of the polyhedron that are incident to $e$. We call a vertex in $f_1$ or $f_2$ {\it opposite to $e$} if it is not incident to $e$. If $f_1$ and $f_2$ are triangles, then there are exactly two vertices opposite to~$e$, see Figure~\ref{fig:oppPoly}.

\begin{figure} \centering
\begin{minipage}[b]{0.41\textwidth}
	\centering
	\includegraphics[height=40mm]{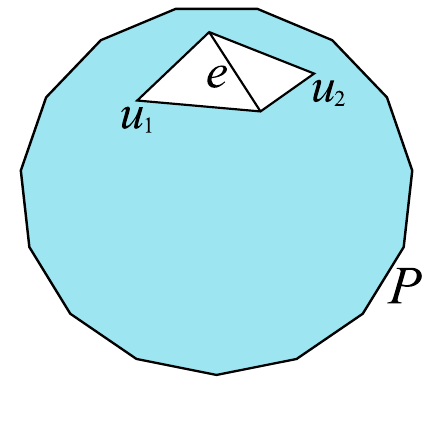}
	\caption{Vertices $u_1$ and $u_2$ are opposite \\ to edge $e$ of polyhedron $P$.}
	\label{fig:oppPoly}
\end{minipage}\hspace{1.15cm}
\begin{minipage}[b]{0.41\textwidth}
	\centering
	\tikz[scale=1.5]{ \def\fl#1#2#3{(0.94868 * #1 cm + 0.47 * #2 cm,
	-0.31623 * #1 cm + 0.67 * #2 cm + 0.9 * #3 cm)}
\def\A{\fl{-1}{-1}{0}}
\def\B{\fl{-1}{1}{0}}
\def\C{\fl{1}{1}{0}}
\def\D{\fl{1}{-1}{0}}
\def\H{\fl{0}{0}{1.41421}}
	
\draw[color=thr] \B -- \C;

\zall 1 \D -- \C -- \H -- cycle;
\zall 3 \A -- \D -- \H -- cycle;
\zall 5 \A -- \B -- \H -- cycle;

\draw \A -- \H;
\draw \B -- \H;
\draw \C -- \H;
\draw \D -- \H;
\draw \B -- \A -- \D -- \C;

\draw \D node[below]{$\ll 2 - \frac{1}{2} - \frac{2}{3} \rr \pi = \frac{5}{6} \pi$};
\draw \A node[left]{$\frac{5}{6} \pi$};
\draw \B node[left]{$\frac{5}{6} \pi$};
\draw \C node[right]{$\frac{5}{6} \pi$};

\draw \H node[above]{\qquad $\ll 2 - \frac{4}{3} \rr \pi = \frac{2}{3} \pi$}; }
     \caption{Gaussian curvature of \\ the vertices of a convex pentahedron.}
     \label{fig:gauscurv}
\end{minipage}
\end{figure}

\begin{definition} \label{def:gauscurv}
	Let $P$ be a convex polyhedron. The \emph{Gaussian curvature} at a vertex $v$ of $P$ equals
	    $\ll 2\pi - \sum_{j=1}^t{\alpha^v_j}\rr$,
	where $t$ is the number of faces of $P$ incident to $v$, and  $\alpha^v_j$ is the angle at $v$ of the $j$-th face incident to $v$.
\end{definition}
	
	Since $P$ is convex, the Gaussian curvature at each vertex of $P$ is non-negative.

\begin{theorem}[Gauss, Bonnet 1848]
	The total sum of the Gaussian curvature of all vertices of a 3D polyhedron $P$ equals $4\pi$.
\end{theorem}	

	For an example, see Figure~\ref{fig:gauscurv} that shows a convex pentahedron and the values of Gaussian curvature at each of its vertices.

\begin{definition} \label{def:gluingDef}
	{\it A gluing $G$} is a collection of polygons $T_1 \ldots T_n$ equipped with an equivalence relation $\sim$ on their border describing how the polygons should be glued to one another.
\end{definition}

\begin{definition} \label{def:polyMetric}
	The {\it polyhedral metric} $M$ of a gluing $G$ is the intrinsic metric of the simplicial complex corresponding to $G$: the distance between two points of the gluing is the infimum of the lengths of the polygonal lines joining the points such that each vertex of it is within one of the polygons $T_1 \ldots T_n$.
\end{definition}

We denote the distance between points $p$, $q$ of $G$ by $|pq|$.

\begin{definition} \label{def:alexCond}
	Gluing $G$ (and the polyhedral metric corresponding to it) is said to satisfy \emph{Alexandrov's conditions} if: \begin{itemize}
	     \item[a)] the topological space produced by $G$ is homeomorphic to a sphere, and
	     \item[b)] the total sum of angles at each of the vertices of $G$ is at most $2\pi$. 
	\end{itemize}
\end{definition}

\begin{theorem}[Alexandrov, 1950,~\cite{alex}]
\label{thm:alexandrov}
	If a gluing $G$ satisfies Alexandrov's conditions then this gluing corresponds to a unique convex polyhedron $\P (G)$: that is, the polyhedral metric of $G$ and the shortest-path metric of the surface of $\P (G)$ are equivalent.
\end{theorem}

Correspondence to a polyhedron discribed in this theorem intuitively means that $\P (G)$ can be glued from polygons of $G$ in accordance with relation~$\sim$. Note that polygons of $G$ need not correspond to faces of $\P (G)$.

Recall that a chord of a polygon $Q$ is any segment connecting two points on the border of $Q$ that lies completely inside $Q$.

\begin{definition} \label{def:netDef} 
	For a polyhedron $P$, a net of $P$ is a gluing $G = (T_1 \ldots T_n, \sim)$ of $P$  together with the  set of chords of the polygons $T_i$ that do not intersect each other except possibly at endpoints. Those chords represent creases, i.e. lines along which $P$ should be folded from this polygon.
\end{definition}

\section{Gluing regular pentagons together}
\label{sec:enum}

In this section, we describe how to enumerate all the edge-to-edge gluings of regular pentagons.

\subsection{How many pentagons can we glue and which vertices can we obtain?}

	Let $P$ be a convex polyhedron obtained by gluing several regular pentagons edge to edge. Vertices of $P$ are clearly vertices of the pentagons. The sum of facial angles around a vertex $v$ of $P$ equals $3\pi/5$ (the interior angle of a regular pentagon) times the number of pentagons glued together at $v$. Since the Gaussian curvature at $v$ is in $(0,2\pi)$, the number of pentagons glued at $v$ can be either one, two, or three. This yields the Gaussian curvature at $v$ to be respectively $7\pi/5$, $4\pi/5$, or $\pi/5$.

	Note that, as opposed to the case of regular hexagons, it is not possible to produce a vertex of curvature $0$ (which would be a flat point on the surface of $P$) by gluing several pentagons. Therefore all the vertices of the pentagons must correspond to vertices of $P$.  

\begin{prop}
\label{prop:upperbound}
	Suppose $P$ is a convex polyhedron obtained by gluing edge-to-edge $N$ regular pentagons.
	Then: (a) $P$ has $2 + 1.5N$ vertices in total. In particular, $N$ must be even.
	(b) $N$ is at most $12$. 
\end{prop}  

\begin{proof}
	From the above discussion, the vertices of $P$ can be subdivided into three types according to their Gaussian curvature:
	(1) the ones of curvature $7\pi/5$,
	(2) $4\pi/5$, and
	(3) $\pi/5$.
	Let us denote the number of vertices type 1, 2 and 3, respectively, as $x, y, z$. Then we have the following system of two equations: 

\[ \begin{cases}
	7x+4y+z=20 \\
	x+2y+3z = 5N
\end{cases} \]

The first equation is implied by the Gauss-Bonnet theorem; the second one is obtained by counting the vertices of pentagons, since each polyhedron vertex of type 1, 2 and 3  corresponds to respectively one, two or three pentagon vertices.

(a) By summing up the equations after multiplying the first one by $0.1$ and the second one by $0.3$, we obtain that $x+y+z = 2 + 1.5N$.  

(b) Since $x,y,z$ are non-negative integers, from the first equation we derive that the maximum number of vertices is obtained when $x = 0, y = 0, z = 20$. This assignment corresponds to $N = 12$ by the second equation. \end{proof}

\subsection{Enumerating all possible gluings.}

\label{sec:method}
We used a computer program to list all the non-isomorphic gluings of this type. Our program is a simple modification of the one that enumerates the gluings of hexagons~\cite{kl17-hex}. The gluings are depicted in Figures \ref{fig:net21}, \ref{fig:net22}, \ref{fig:net41}, \ref{fig:net42}, \ref{fig:net43}, \ref{fig:net6}, \ref{fig:net8}, \ref{fig:net12}.

\begin{figure}[h] \centering
\begin{subfigure}[b]{0.4\columnwidth} \centering
	\tikz[scale=1.2]{
	\input{graphics/inscope-1}
	}
     \caption{$\P_{2,1}$}
     \label{fig:1}
\end{subfigure}
~~
\begin{subfigure}[b]{0.4\columnwidth} \centering
	\input{graphics/2}
	\caption{$\P_{2,2}$}
	\label{fig:2}
\end{subfigure}
\\
\begin{subfigure}[b]{0.4\columnwidth} \centering
	\includegraphics{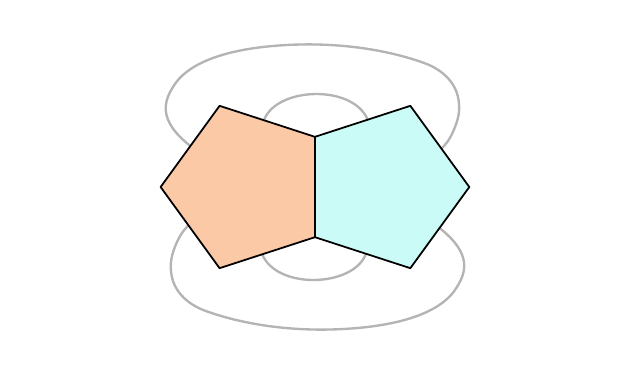}
	\caption{Net of $\P_{2,1}$}
	\label{fig:net21}
\end{subfigure}
~~
\begin{subfigure}[b]{0.4\columnwidth} \centering
	\includegraphics{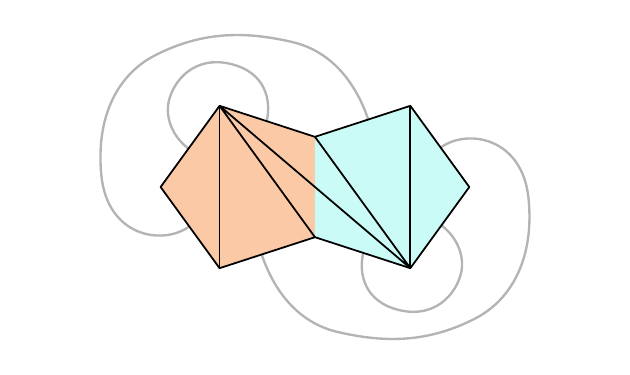}
	\caption{Net of $\P_{2,2}$}
	\label{fig:net22}
\end{subfigure}
	\caption{Polyhedra glued from two regular pentagons and their nets. Here and further black lines are creases along which the polyhedron is folded. Dark red lines always denote borders between the polygons of the gluing.}
\end{figure}

\section{A complete list of all shapes obtained by gluing pentagons}
\label{section:complete}

Below is the list of all polyhedra that can be obtained by gluing regular pentagons. For those polyhedra that are simplicial, their graph structure is confirmed by applying method of Section~\ref{section:algorithmic}, for the others the proof is geometric and is done in Section~\ref{section:geom}.

\begin{itemize}
    \item 2 pentagons:\vspace{-3.5mm}
    \begin{itemize}
        \item doubly-covered regular pentagon, see Figures~\ref{fig:1},~\ref{fig:net21}.
        \item simplicial hexahedron with 5 vertices (3 vertices of degree 4, and 2 vertices of degree 3), see Figures~\ref{fig:2},~\ref{fig:net22}.
    \end{itemize}
\end{itemize} \vspace{-3.5mm}

\begin{figure}[h] \centering
\begin{subfigure}[b]{0.28\columnwidth} \centering
	\input{graphics/4-1}
	\caption{$\P_{4,1}$}
	\label{fig:4-1}
\end{subfigure}
~~
\begin{subfigure}[b]{0.24\columnwidth} \centering
	\input{graphics/4-2}
	\caption{$\P_{4,2}$}
	\label{fig:4-2}
\end{subfigure}
~~
\begin{subfigure}[b]{0.34\columnwidth} \centering
	\input{graphics/4-3}
	\caption{$\P_{4,3}$}
	\label{fig:4-3}
\end{subfigure}
\\
\begin{subfigure}[b]{0.28\columnwidth} \centering
	\includegraphics{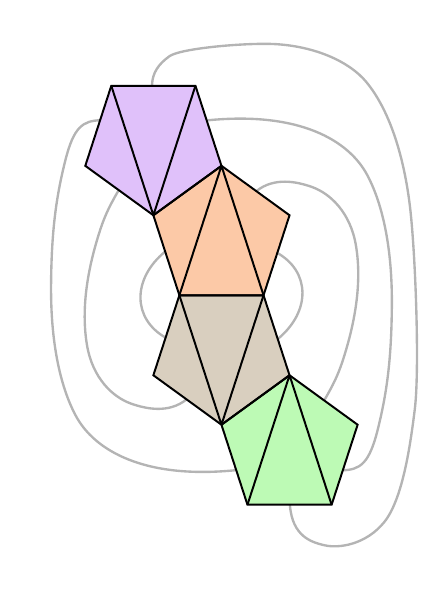}
	\caption{Net of $\P_{4,1}$}
	\label{fig:net41}
\end{subfigure}
~~
\begin{subfigure}[b]{0.28\columnwidth} \centering
	\includegraphics{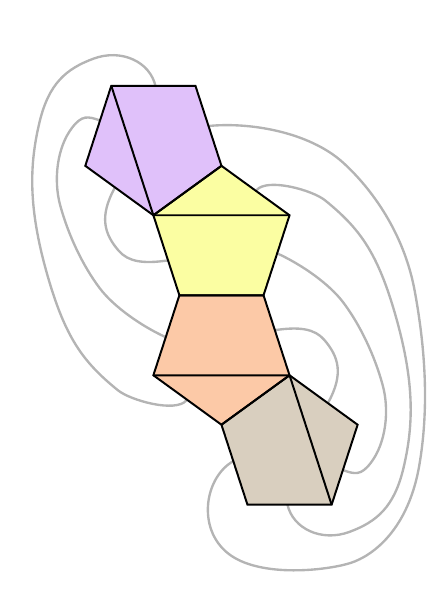}
	\caption{Net of $\P_{4,2}$}
	\label{fig:net42}
\end{subfigure}
~~
\begin{subfigure}[b]{0.28\columnwidth} \centering
	\includegraphics{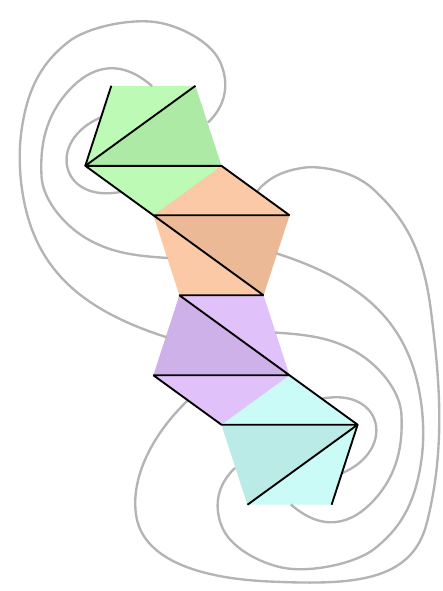}
	\caption{Net of $\P_{4,3}$}
	\label{fig:net43}
\end{subfigure}
	\caption{Polyhedra glued from four regular pentagons and their nets.}
\end{figure}

\begin{itemize}
    \item 4 pentagons:\vspace{-3.5mm}
    \begin{itemize}
        \item simplicial dodecahedron with 8 vertices (2 vertices of degree 5 and 6 vertices of degree 4), see Figures~\ref{fig:4-1},~\ref{fig:net41}.
        \item octohedron with 8 vertices (4 vertices of degree 4 and 4 vertices of degree 3) and 4 quadrilateral and 4 triangular faces.  It is a truncated biprism, see Figures~\ref{fig:4-2},~\ref{fig:net42}.
        \item hexahedron with 8 vertices each of degree 3 and 6 quadrilateral faces. This is a parallelepiped, see Figures~\ref{fig:4-3},~\ref{fig:net43}.
    \end{itemize}
\end{itemize}

Note that $\P_{4,1}$, $\P_{4,2}$, $\P_{4,3}$ can be glued from a single common polygon by altering the relation $\sim$.

\begin{figure} \centering
\begin{subfigure}[b]{0.24\columnwidth}
	\input{graphics/6}
	\caption{$\P_6$}
	\label{fig:6}
\end{subfigure}
~~
\begin{subfigure}[b]{0.26\columnwidth}
	\input{graphics/8}
	\caption{$\P_8$}
	\label{fig:8}
\end{subfigure}
~~
\begin{subfigure}[b]{0.28\columnwidth}
\begin{center}
	\tikz[scale=1]{
	\input{graphics/inscope-12}
	}
\end{center}
	\caption{$\P_{12}$}
	\label{fig:12}
\end{subfigure}
\\ \ \\
\begin{subfigure}[b]{0.45\columnwidth} \centering
	\includegraphics[scale=1.08]{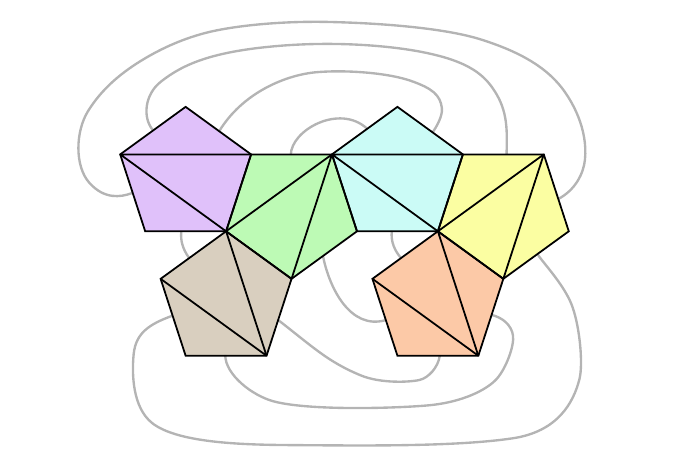}
	\caption{Net of $\P_6$}
	\label{fig:net6}
\end{subfigure}
~~
\begin{subfigure}[b]{0.45\columnwidth} \centering
	\includegraphics[scale=1.08]{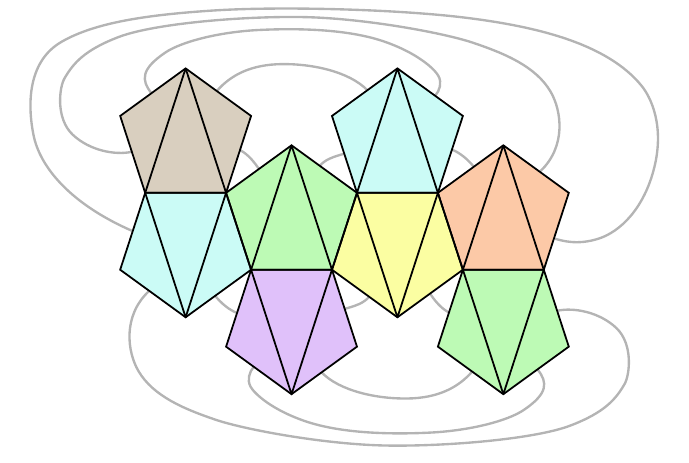}
	\caption{Net of $\P_8$}
	\label{fig:net8}
\end{subfigure}
\\ \ \\
\begin{subfigure}[b]{0.7\columnwidth} \centering
	\includegraphics[scale=0.95]{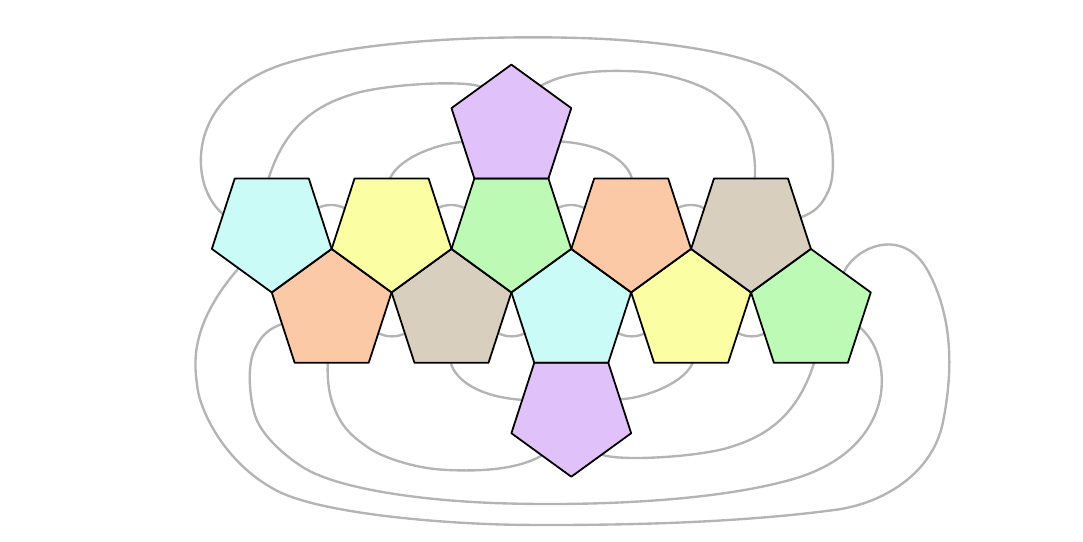}
	\caption{Net of $\P_{12}$}
	\label{fig:net12}
\end{subfigure}
	\caption{Polyhedra glued from six or more regular pentagons and their nets.}
\end{figure}

\begin{itemize}
    \item 6 pentagons: simplicial decaoctohedron (18-hedron) with 11 vertices (5 vertices of degree 6, 6 vertices of degree 4), see Figures~\ref{fig:6},~\ref{fig:net6}.
    \item 8 pentagons: simplicial icositetrahedron (24-hedron) with 14 vertices (2 vertices of degree 6, 12 vertices of degree 5), see Figures~\ref{fig:8},~\ref{fig:net8}.
    \item 12 pentagons: regular dodecahedron with 20 vertices of degree 3 and 12 pentagonal faces, see Figure~\ref{fig:12},~\ref{fig:net12}. 
\end{itemize}

We now proceed with a description of how to determine the graph structures of the polyhedra in this list. We separately confirm the presence of the edges (Section~\ref{section:algorithmic}) and prove that no additional edges are present in the quadrilateral faces of $\P_{4,2}$ and $\P_{4,3}$ (Section~\ref{section:geom}). 

\section{An algorithmic method to verify the graph structure of a glued polyhedron} \label{section:algorithmic}

Consider a polyhedral metric $M$ that satisfies the Alexandrov's conditions and thus corresponds to a unique polyhedron $\P$. Suppose we have a polyhedron $P$ that approximates $\P$. That is, vertices of $P$ are in one-to-one correspondence with the cone points of $M$ (and thus with the vertices of $\P$). In this section we show how to check whether the graph structure of $\P$ contains all the edges of $P$. 

We will be using the following notation: $v_1$, $v_2$, $v_3$,~$\ldots$ for the vertices of $\P$; $u_1$, $u_2$, $u_3$,~$\ldots$ for the corresponding  vertices of $P$; $V$, $E$, $F$ for the number of vertices, edges and faces of $P$ respectively; $\mathcal D$ for the maximum degree of a vertex of $P$; $L$ for the length of the longest edge of $P$; $\ball{u}$ for the ball in $\br^3$ of radius $r$ centered at the point $u$.

We also know the lengths of edges and distances between vertices of $\P$ since those are lengths of shortest paths between cone points of metric $M$. Let the \emph{discrepancy} of an edge $u_iu_j$ of $P$ be the absolute value of the difference between the length of that edge and the distance between the corresponding vertices $v_i$  and $v_j$ of $\P$. Let \emph{maximum edge discrepancy}  $\mu$ of $P$ be the maximum discrepancy for all edges of $P$.

Similarly, for any facial angle $u_ju_iu_k$ of $P$, let \emph{discrepancy} of this angle be the absolute value of the difference between the values of $u_ju_iu_k$ and of the angle between the corresponding shortest paths in $\P$; let the maximum angle discrepancy $\gamma$ of $P$ be the maximum discrepancy for all the facial angles of $P$.

We base our check on the following theorem.

\begin{theorem}
\label{thm:precision}
	Suppose $\mu$ is the maximum edge discrepancy between $P$ and $\P$, $\gamma$ is the maximum angle discrepancy between $P$ and $\P$, $\mathcal D$ is the maximum degree of a vertex of $P$. If $\mathcal D \gamma < \pi / 2$, then each vertex of $\P$ lies within an $r$--ball centered at the corresponding vertex of $P$, where
\begin{equation}
	r = E^2 \cdot  L \cdot 2 \sin ( \mathcal D \gamma / 2 ) + E \mu.
\end{equation} \end{theorem}

We defer its proof to the Section~\ref{section:precis}, and for now we focus on describing our check, using the theorem as a black box.

Let $u_iu_j$ be an edge of $P$ and let $u_a$, $u_b$ be the two vertices of $P$ opposite to the edge $u_iu_j$ (see Figure~\ref{fig:oppPoly}). We want to check that there does not exist a plane intersecting all four $r$--balls centered at $u_i$, $u_j$, $u_a$, $u_b$ respectively.

Assume without loss of generality that the plane passing through $u_a$, $u_i$, $u_j$ is not vertical and that $P$ lies below that plane (otherwise apply a rigid transformation to $P$ so that it becomes true).
Note that we always can do this since $P$ is convex.

Consider three planes $\Pi_1$, $\Pi_2$, $\Pi_3$ tangent to $\ball{u_i}$, $\ball{u_j}$, $\ball{u_a}$ such that:

\begin{itemize}
	\item $\Pi_1$ is below $\ball{u_i}$, $\ball{u_j}$ and above $\ball{u_a}$,	
	\item $\Pi_2$ is below $\ball{u_i}$ and above $\ball{u_j}$, $\ball{u_a}$,	
	\item $\Pi_3$ is below $\ball{u_j}$ and above $\ball{u_i}$, $\ball{u_a}$.
\end{itemize}

\begin{figure}[h] \centering \def\scopescale{0.77}
	\input{figs/tang_plane}
	\caption{Plane $\Pi_1$ tangent to $\ball{u_i}$, $\ball{u_j}$, $\ball{u_a}$.}
	\label{fig:tangPlane}
\end{figure}

\begin{theorem} \label{thm:code-whattocheck}
	If $u_b$ lies below $\Pi_1$, $\Pi_2$ and $\Pi_3$ and the distance from $u_b$ to each of the planes $\Pi_1$, $\Pi_2$ and $\Pi_3$ is greater than $r$, then there must be the edge $v_iv_j$ in $\P$.
\end{theorem}

An example can be seen on Figure~\ref{fig:tangPlane}: plane $\Pi_1$ is tangent to $\ball{u_i}$, $\ball{u_j}$, $\ball{u_a}$. Point $u_{b,1}$ is below $\Pi_1$, and point $u_{b,2}$ is above $\Pi_1$, the distance from each of the points to $\Pi_1$ is greater than $r$.

To prove this theorem, we need the following lemma.

\begin{lemma} \label{lm:ballsPlaneCase}
	Given two disks $\ball\ule$, $\ball\uri$ in $\br^2$; points $\ule$, $\uri$ lie on $x$ axis. Given a point $u$, $x_u > x_{\uri}$, $y_u < 0$.
	If $u$ lies below the common tangent of the disks that is above $\ball\ule$ and below $\ball\uri$, than
	there is no line passing through $\ball\ule$, $\ball\uri$, and $u$.
\end{lemma}

The example for this lemma can be seen in Figure~\ref{fig:ballsPlaneCaseEx}. Point $u_1$ is above the tangent, so there may be a line passing through it and the two disks. Point $u_2$ is below the tangent, so no lines through $\ball\ule$, $\ball\uri$, $u$ are possible. \begin{figure}
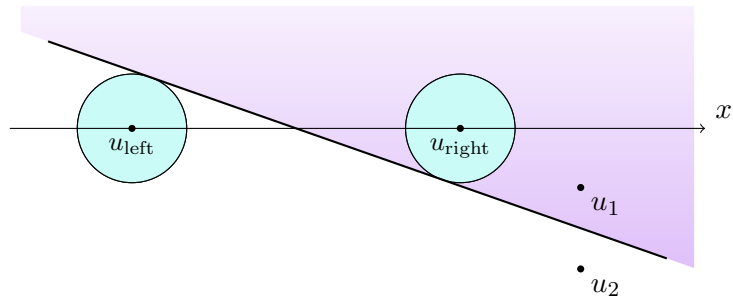
 \centering
\tikz[scale=0.36]{
	\draw[white] (- 2.1 * 5.3333 cm - 1.5 cm,0) -- (0,0);
	\shade[opacity=0.55, bottom color=fill5, top color=fill5!24!white]
		(-1.7 * 5.3333 cm - 1 cm, 1.7 * 1.8856 cm + 0.3536 cm)
		-- (2.54 * 5.3333 cm + 1cm, -2.54 * 1.8856 cm - 0.3536 cm)
		-- (2.54 * 5.3333 cm + 1cm, 4.5)
		-- (-1.7 * 5.3333 cm - 1 cm, 4.5) -- cycle;
     \foreach \r in {-6,6} {
	\filldraw[fill=white] (\r,0) circle[radius=2cm];
	\filldraw[fill=fill4,fill opacity=0.55] (\r,0) circle[radius=2cm];
	\fill (\r,0) circle[radius=1.3mm];
     }
	\draw[thick] (-1.7 * 5.3333 cm, 1.7 * 1.8856 cm) --
		(2.54 * 5.3333 cm, -2.54 * 1.8856 cm);
	\draw[->] (-1.7 * 5.3333 cm - 1.4 cm,0) -- (2.54 * 5.3333 cm + 1.4 cm,0)
		node[anchor=south west]{$x$};
     \fill (1.95 * 5.3333 cm, -1.95 * 1.8856 cm)
	++(0,1.5) circle[radius=1.3mm]
		node[anchor=north west]{$u_1$}
	++(0,-3) circle[radius=1.3mm]
		node[anchor=north west]{$u_2$};
     \draw (-6,0) node[below]{\small $\ule$}
	(6,0) node[below]{\small $\uri$};
}
	\caption{An example for Lemma~\ref{lm:ballsPlaneCase}.}
	\label{fig:ballsPlaneCaseEx}
\end{figure}

\begin{proof}
Consider the set of points in $\br^2$ covered by all lines passing through $\ball\ule$, $\ball\uri$. We are looking for the lower border of it which corresponds to the lowest line passing through these disks. \begin{figure}
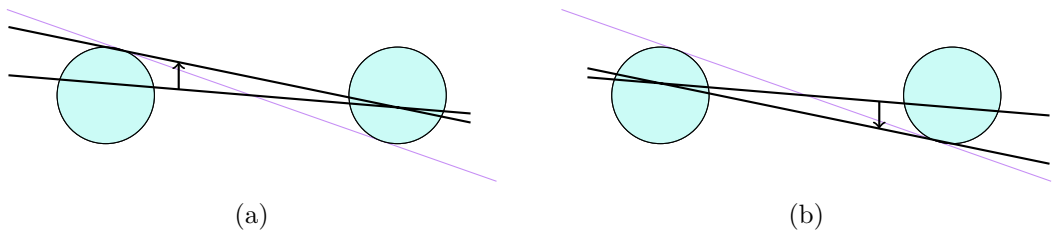
 \centering
\begin{subfigure}[t]{0.44\textwidth} \centering
\tikz[scale=0.32,rotate=180]{
	\draw[draw=fill5]
		(-1.7 * 5.3333 cm - 1 cm, 1.7 * 1.8856 cm + 0.3536 cm)
		-- (1.7 * 5.3333 cm + 1cm, -1.7 * 1.8856 cm - 0.3536 cm);
     \foreach \r in {-6,6} {
	\filldraw[fill=white] (\r,0) circle[radius=2cm];
	\filldraw[fill=fill4,fill opacity=0.55] (\r,0) circle[radius=2cm];
     }
	\draw[thick] (-6,0.5) -- ++(-3, 0.25) (-6,0.5) -- ++(12,-1) -- ++(4,-0.3333)
		(-6,0.5) -- ++(-3,0.25*2.5 cm) (-6,0.5) -- ++(12,-2.5) -- ++(4,-0.3333*2.5 cm);
	\draw[thick,->] (-6,0.5) ++(0.75*12,-0.75) -- ++(0,-0.75*1.5 cm);
} \caption{ }
\label{fig:tangRaiseA}
\end{subfigure}\hspace{0.5cm}
\begin{subfigure}[t]{0.44\textwidth} \centering
\tikz[scale=0.32]{
	\draw[draw=fill5]
		(-1.7 * 5.3333 cm - 1 cm, 1.7 * 1.8856 cm + 0.3536 cm)
		-- (1.7 * 5.3333 cm + 1cm, -1.7 * 1.8856 cm - 0.3536 cm);
     \foreach \r in {-6,6} {
	\filldraw[fill=white] (\r,0) circle[radius=2cm];
	\filldraw[fill=fill4,fill opacity=0.55] (\r,0) circle[radius=2cm];
     }
	\draw[thick] (-6,0.5) -- ++(-3, 0.25) (-6,0.5) -- ++(12,-1) -- ++(4,-0.3333)
		(-6,0.5) -- ++(-3,0.25*2.5 cm) (-6,0.5) -- ++(12,-2.5) -- ++(4,-0.3333*2.5 cm);
	\draw[thick,->] (-6,0.5) ++(0.75*12,-0.75) -- ++(0,-0.75*1.5 cm);
}  \caption{ }
\label{fig:tangRaiseB}
\end{subfigure}
	\caption{Common tangent of disks is lower than any line passing through them.}
\end{figure}

Consider a line passing through the disks. If it is not tangent to $\ball \ule$ from above, it can be made lower by raising its intersection with $\ball\ule$, see Figure~\ref{fig:tangRaiseA}. If it is not tangent to $\ball\uri$ from below, it also can be made lower by lowering its intersection with $\ball\uri$, see Figure~\ref{fig:tangRaiseB}.

Therefore, any line passing through $\ball\ule$, $\ball\uri$ is higher than the common tangent of these disks when $x > x_{\uri}$.
\end{proof}

\begin{proof}[Proof of Theorem~\ref{thm:code-whattocheck}]
We can assume that points $u_i$, $u_j$ lie on $y$ axis, see Figure~\ref{fig:tangPlane}. For each pair $(x,y)$ we want to find minimum $z$ such that there is a plane passing through $\ball{u_i}$, $\ball{u_j}$, $\ball{u_a}$, and $(x,y,z)$. Let us consider three cases: (1) $y_{u_i} \le y \le y_{u_j}$, (2) $y \le y_{u_i}$, (3) $y_{u_j} \le y$.

Consider case 1. Project everything on plane $y=0$. The projections of $\ball{u_i}$ and $\ball{u_j}$ coincide, and a plane $\l$ passing through these disks can be lowered by matching the projections of its intersections with the disks, thus making projection of $\l$ a line. Now we can apply Lemma~\ref{lm:ballsPlaneCase} to the projection to get plane $\Pi_1$ from the statement of the Theorem.

Consider case 2. Project everything on a plane orthogonal to the segment $u_ju_a$. Using similar argument, applying Lemma~\ref{lm:ballsPlaneCase} we get plane $\Pi_2$ from the statement. Case 3 is symmetric to case 2 and gives us plane $\Pi_3$.

Therefore, all points of $\ball{u_b}$ should lie below the planes $\Pi_1$, $\Pi_2$, $\Pi_3$, which yields the condition of distance between $u_b$ and the planes being greater than $r$.
\end{proof}

The check suggested in Theorem~\ref{thm:code-whattocheck} requires $O(1)$ time, and has to be performed once for every edge $u_iu_j$ of $P$. This implies the following.

\begin{theorem} \label{thm:checkGraphStructure}
	Given a polyhedral metric $M$ satisfying Alexandrov's conditions and an approximation $P$ for the polyhedron $\P$ that corresponds to $M$, there is a procedure to verify for each edge of $P$ if it is present in $\P$. The procedure answers ``yes'' only for those edges that are present in $\P$, and it answers ``inconclusive'' if the approximation $P$ is not precise enough. The procedure requires time $\mathcal O(E)$. 
\end{theorem}

Inconclusive answers occur if a plane exists that intersects all four $r$--balls even though there is an edge connecting two of the vertices. In such case, precision has to be increased by replacing $P$ with a polyhedron that has smaller discrepancy in edge lengths and values of angles and repeating the procedure.

Theorem~\ref{thm:checkGraphStructure} yields that if $\P$ is simplicial we can in time $\mathcal O(E)$ verify whole its graph structure without any additional effort. However, if there are faces in $\P$ with four or more vertices, the absence of the edges that are diagonals of these faces has to be proved, which requires some creativity. For non-simplicial shapes glued from pentagons such proofs are given in Section~\ref{section:geom}.

To obtain polyhedron $P$ one can use the algorithm developed by Kane et al.~\cite{kpd09-approx} or the one by Bobenko, Izmestiev~\cite{boben}. Each of them outputs a polyhedron $P$ which is an approximation of $\P$. In this work we used the implementation of the latter presented by Sechelmann~\cite{sech}. It gave us approximation with $\mu \sim 10^{-7}$, $\gamma \sim 10^{-6}$, $L \sim 2.5$. These parameters allowed for $r \sim 10^{-3}$, which was enough to verify the presence of all the suggested edges.

To do so, we developed a program that checks the condition of Theorem~\ref{thm:code-whattocheck}. Its source code can be found in our bitbucket
repository\footnote{\url{bitbucket.org/boris-a-zolotov/diplomnaia-rabota-19/src/master/praxis/haskell}}.
 
\subsection{Proof of Theorem~\ref{thm:precision}} \label{section:precis}

We now proceed with the proof of Theorem~\ref{thm:precision}. To prove it, we need the following lemma.

\begin{lemma}
\label{lm:singlesegm}
Let $pq$, $pq'$ be line segments in $\br^3$, $|pq| = \ell$.
If there are two real numbers $\varepsilon$, $\theta$ with $\varepsilon > 0$
and $0 < \theta < \tfrac{\pi}{2}$ such that
	\[ \ell-\varepsilon \le |pq'| \le \ell+\varepsilon
	\text{\quad and\quad}
	\magl qpq' \le \theta, \]
then\quad $|qq'| \le 2 \ell \sin\frac{\theta}{2} + \varepsilon.
	\quad\refstepcounter{equation}\hfill(\theequation)\label{eq:singlesegm}$
\end{lemma}

\begin{proof}

	$pq'$ can be obtained from $pq$, as shown in Figure~\ref{fig:edgeOffset}, by a composition $\rho \circ \tau$ of
	
\begin{itemize}
	\item[(1)] rotation $\rho$ around $p$ by an angle at most $\theta$,
	\item[(2)] homothety $\tau$ with center $p$ and ratio $\lambda$, where $\lambda$ is some real number with 
	$\frac{\ell-\varepsilon}{\ell} \le \lambda \le \frac{\ell+\varepsilon}{\ell}$.
\end{itemize}

\begin{figure}[h]
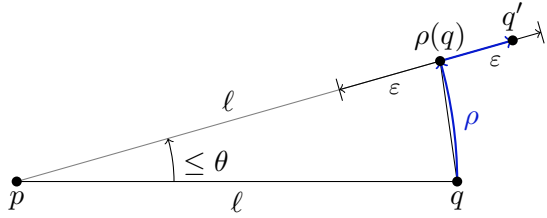
 \centering
	\tikz[scale=1.38]{ \def\ptmark#1{\fill (#1) circle[radius=0.48mm];}
	\draw (0,0) node[below]{$p$}
	    -- (0:4.2) node[below]{$q$}
	    -- (16:4.2) node[above]{$\rho(q)$};
	\draw (0:2.1) node[below]{$\ell$}; \draw (16:2.1) node[above]{$\ell$};
	\draw[color=gray] (0,0) -- (16:4.2);
	\draw[->] (1.5,0) arc (0:7:1.5) node[right]{$\le\theta$} arc(7:16:1.5);
	\draw (16:4.92) node[above]{$q'$};
    \begin{scope}[rotate=16]
	\draw[<->] (3.2,0) -- (5.2,0);
	\foreach \x in {3.2,5.2} {\draw (\x, -0.11) -- (\x, 0.11);}
	\foreach \x in {3.75,4.75} {\draw (\x, 0) node[below]{\footnotesize $\varepsilon$};}
    \end{scope}
	\draw[color=highlight,->,thick] (4.2,0)
	    arc (0:8:4.2) node[right]{$\rho$} arc (8:16:4.2);
	\draw[color=highlight,->,thick] (16:4.2) -- (16:4.92);		  
	\ptmark{0,0} \ptmark{4.2,0}
	\ptmark{16:4.2} \ptmark{16:4.92}
}
	\caption{After a segment is rotated by at most $\theta$ and its length changed by at most $\varepsilon$, its endpoint $q$ moves by at most $\ell \cdot 2\sin\tfrac{\theta}{2} + \varepsilon$.}
	\label{fig:edgeOffset}
\end{figure}
	
First, it is clear that $\left| \rho(q),\ \tau(\rho(q)) \right|\le\varepsilon$, since $\tau$ is defined so as to add not more than $\varepsilon$ to a segment of length $\ell$. Now we estimate $\dist (q, \rho(q))$. It is at most $\ell \cdot 2\sin ( \theta / 2)$, which is the length of the base of an isosceles triangle with sides equal to $\ell$ and angle at the apex $\theta$.
	
Combining the above estimations with the triangle inequality concludes the proof. \end{proof}

\begin{figure}
	\centering
	\includegraphics[scale=1.1]{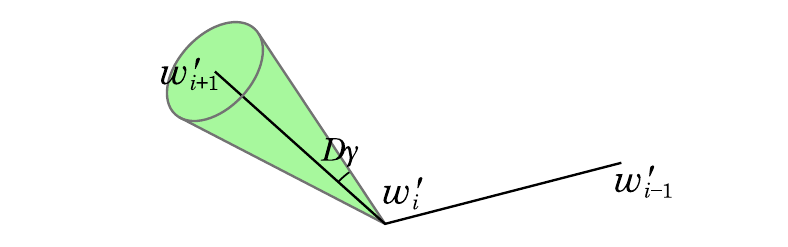}
	\caption{The angle between the edge of $\P$ and the edge of $P$ is less than $\mathcal D\gamma$.}
	\label{fig:angleOffset}
\end{figure}

\begin{proof}[Proof of Theorem~\ref{thm:precision}]
Place $P$ and $\P$ in such a way that
\begin{enumerate}
	\item a pair of their corresponding vertices, $u_1$ in $P$ and $v_1$ in $\P$, coincide,
	\item a pair of corresponding edges, $e'$ incident to $u_1$ in $P$ and $e$ incident to $v_1$ in $P$, lie on the same ray, and
	\item a pair of corresponding faces, $f'$ in $P$ incident to $u_1$ and $e'$ and $f$ in $\P$ incident to $v_1$ and $e$, lie on the same half-plane.
\end{enumerate}

Consider a pair of corresponding vertices, $u$ in $P$ and $v$ in $\P$. In order to estimate $|uv|$ consider a shortest path $\pi_1 = u_1w'_1w'_2\ldots w'_ku$ in the graph structure of polyhedron $P$. It is comprised of edges of $P$ and is not the geodesic shortest path from $u_1$ to $u$. Vertices of $\pi_1$ correspond to the vertices of another path $\pi_2 = v_1w_1w_2\ldots w_kv$ in $\P$. Since $\pi_1$ is a simple path, it contains at most $E$ edges and therefore its total length is at most $EL$.

We now focus on the paths themselves, not on the polyhedra. Path $\pi_2$ can be obtained from $\pi_1$ by a sequence of changes of edge directions
	(see Figures~\ref{fig:angleOffset},~\ref{fig:anglePathOffset})
and edge lengths (see Figure~\ref{fig:edgePathOffset}). Let us estimate by how much endpoint $u$ of path $\pi_1$ can move when this sequence of changes is applied.

Denote $w_0' \coloneqq u_1$, $w_0 \coloneqq v_1$ and assume that for each $j = 1, \ldots, i$ edge $w_{j-1}' w_j'$ is parallel to $w_{j-1} w_j$. Then, by the triangle inequality, the angle $\alpha$ between $w'_i w'_{i+1}$ and $w_i w_{i+1}$ is at most $\mathcal D\gamma$, see Figure~\ref{fig:angleOffset}. Rotate the path $w_i' \ldots w_k' u$ around $w_i'$ by angle $\alpha$ so $w_i' w_{i+1}'$ and $w_i w_{i+1}$ become parallel.

Distance $|w_i' u|$ is at most $EL$, so, by Lemma~\ref{lm:singlesegm}, every time we apply such rotation, the endpoint $u$ of path $\pi_1$ moves by at most $E L \cdot 2 \sin (\mathcal D \gamma / 2)$. Since there are at most $E$ vertices in the path and $E$ rotations are applied, the endpoint $u$ moves by at most
\begin{equation}
	\label{eq:anglePathOffset}
	E^2 \cdot  L \cdot 2 \sin\ll \frac{\mathcal D \gamma}{2} \rr.
\end{equation}

\begin{figure}
	\centering
	\begin{subfigure}[b]{0.44\columnwidth}
		\centering
		\includegraphics[scale=1.15]{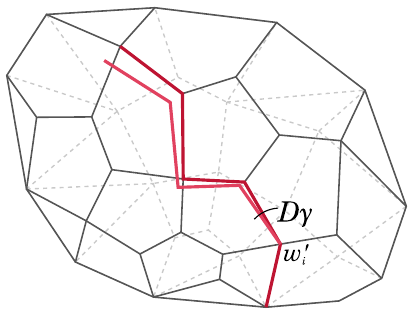}
		\caption{\ }
		\label{fig:anglePathOffset}
	\end{subfigure}
~
	\centering
	\begin{subfigure}[b]{0.44\columnwidth}
		\centering
		\includegraphics[scale=1.15]{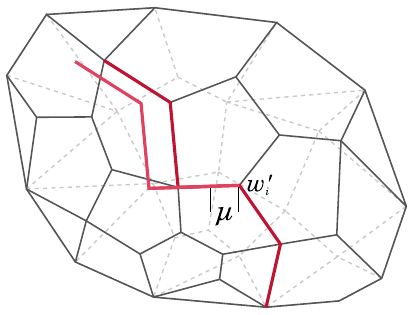}
		\caption{\ }
		\label{fig:edgePathOffset}
	\end{subfigure}
\caption{Illustration for the proof of Theorem~\ref{thm:precision}:
	{\itshape (a)}~Rotation by the angle less than $\mathcal D\gamma$
	is applied to the path $w'_i\ldots w'_k$. {\itshape (b)}~The edge
	$w'_iw'_{i+1}$ is being lengthened or shortened by not more
	than $\mu$. \vspace{-3mm}}
\end{figure}

Now that the directions of all the edges in path $\pi_1$ coincide with the directions of the edges in path $\pi_2$, we can make the lengths of corresponding edges match. If the length of a single edge of a path in $P$ is changed by at most $\mu$, and other edges are not changed (as shown in Figure~\ref{fig:edgePathOffset}), then the end of the path also moves by not more than $\mu$. Therefore after we adjust the length of all the edges, the endpoint $u$ of path $\pi_1$ moves by at most 
	$E \cdot \mu.~\refstepcounter{equation}\hfill(\theequation)\label{eq:edgePathOffset}$

Combining~(\ref{eq:anglePathOffset}) and~(\ref{eq:edgePathOffset}) implies that in total point $u$ moved by at most
\begin{equation}
	E^2 \cdot L \cdot 2 \sin ( \mathcal D \gamma / 2 ) + E \mu.
\end{equation}

This completes the proof. \end{proof}

\section{Geometric methods to determine graph structure} \label{section:geom}

In this section we give the last part of the proof that the polyhedra corresponding to the  gluings listed in Section~\ref{section:complete} have the same graph structure as the polyhedra listed in the same section. That is, we prove that quadrilateral faces of $P_{4,2}$, $P_{4,3}$ correspond to quadrlateral faces of $\P_{4,2}$, $\P_{4,3}$, i. e., that certain edges are not present in $\P_{4,2}$, $\P_{4,3}$.

\subsection{Quadrilateral faces of $\P_{4,2}$}

Recall that  $\P_{4,2}$ is the polyhedron that corresponds to the gluing $G_{4,2}$ (see Figure~\ref{fig:net42}). Let $A, B, \ldots, H$ denote the vertices of $G_{4,2}$, see Figure~\ref{fig:4-2razv}. We have already established by the methods of Section~\ref{section:algorithmic} that $\P_{4,2}$ has edges that are shown in the net on Figure~\ref{fig:net42} (black lines). We now prove the following.

\begin{theorem} \label{thm:4-2symm}
	For the polyhedron $\P_{4,2} = \P(G_{4,2})$, each of the 4-tuples of vertices $(G,H,C,D)$, $(A,B,H,G)$, $(E,F,C,B)$, $(A,D,F,E)$ forms a quadrilateral face of $\P_{4,2}$. 
\end{theorem}

\begin{proof}

\begin{figure}[h]
	\input{figs/4-2symm}
\end{figure} 

\begin{figure}[h]
	\centering
	\includegraphics[width=8.2cm]{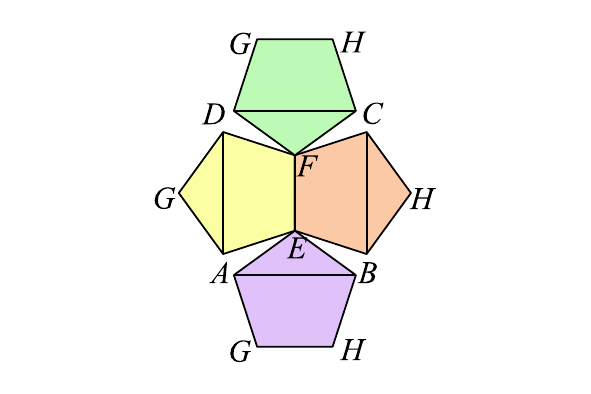}
	\caption{The net of $\P_{4,2}$.}
	\label{fig:4-2razv}
\end{figure}

	Observe first that there are two vertical planes such that $\P_{4,2}$ is symmetric with respect to both of them: (1) the plane $\zeta$ that passes through edge $GH$ (the common side of two pentagons), and the midpoints $M_1$, $M_2$, $M_3$ of edges $AD$, $EF$, $BC$ respectively (see Figure~\ref{fig:4-2symm}); and (2) the plane $\zeta'$ that passes through edge $EF$ and the midpoints of edges $AB$, $GH$ and $DC$. Indeed, polyhedron $\P_{4,2}$ is symmetric with respect to plane $\zeta$, since the segment $HM_2$ cuts in half the pentagon $EFCHB$ (colored orange in Figures~\ref{fig:4-2symm} and~\ref{fig:4-2razv}), and so does the segment $GM_2$ does with pentagon $FEAGD$ (colored yellow in Figures~\ref{fig:4-2symm} and~\ref{fig:4-2razv}). The argument for the plane $\zeta'$ is analogous. 

	Suppose for the sake of contradiction that $BF$ is an edge of $\P_{4,2}$. Then segment $EC$ must also be an edge due to the symmetry with respect to plane $\zeta$. However, segments $BF$ and $EC$ cross inside the pentagon $EFCHB$ and thus cannot be both the chords of the net of $\P_{4,2}$. We arrive to a contradiction. By the same argument $EC$ cannot be an edge of $\P_{4,2}$. Therefore $EFCB$ is a quadrilateral face of $\P_{4,2}$.
	
	The existence of quadrilateral faces $GHCD,\ ABHG,\ ADFE$ is implied by a symmetric argument. This completes the proof. \end{proof}

\subsection{Quadrilateral faces of $\P_{4,3}$}

Polyhedron  $\P_{4,3}$ is the polyhedron that corresponds to the gluing $G_{4,3}$ (see Figure~\ref{fig:net43}). Again let $A, B, \ldots H$ denote the vertices of $G_{4,3}$, see Figure~\ref{fig:4-3razv}. The chords shown in the net on Figure~\ref{fig:net43} (black lines) are already proven to be corresponding to the edges of  $\P_{4,3}$. We now prove the following.

\begin{theorem} \label{thm:4-3symm}
	For the polyhedron  $\P_{4,3} = \P(G_{4,3})$, each of the 4-tuples of vertices $(E,A,B,F)$, $(E,A,D,H)$, $(C,G,F,B)$, $(C,G,H,D)$, $(A,B,C,D)$, $(E,F,G,H)$  
	forms a quadrilateral face of $\P_{4,3}$. In particular, each of these faces is a parallelogram. 
\end{theorem}

\begin{proof}

\begin{figure}[h]
	\input{figs/4-3parall}
\end{figure}

\begin{figure}[h]
	\centering
	\includegraphics[width=8.2cm]{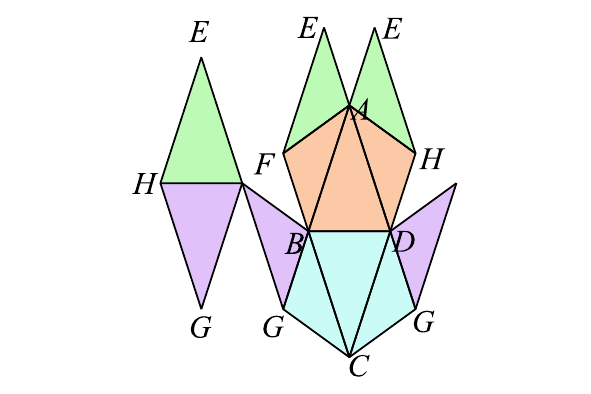}
	\caption{The net of $\P_{4,3}$.}
	\label{fig:4-3razv}
\end{figure}

We show that there is a convex polyhedron with the net as  in  Figure~\ref{fig:4-3razv} that satisfies the claim. By Alexandrov's theorem such polyhedron is unique and is exactly $\P_{4,3}$.  

	The pentagon $EAFHA$ (colored green in Figure~\ref{fig:4-3razv}) is folded along its diagonals $EF$ and $EH$ and glued along its edge $EA$. We use one degree of freedom to place it so that it is symmetric with respect to the plane through  $EAM_1$, where $M_1$ is the midpoint of $HF$. Let us now take another pentagon $AFBDH$ and glue one of its vertices to $A$. Place this pentagon in a way that the plane $ADB$ is parallel to the plane $EHF$ (see the orange pentagon in Figure~\ref{fig:4-3razv}). Now we glue these two pentagons along the edges $AF$ and $AH$ without changing the position of the triangle $ADB$. Since $\magl FEA + \magl EAF + \magl FAB = \pi$, the points $E, A, B, F$ are coplanar and form a parallelogram. By analogous arguments, $EADH$, $CGFB$, and $CGHD$ are parallelograms as well.
	
	It is easy to see that the shape we just obtained by gluing the pentagons $EAFHA$ and $AHDBF$ is still symmetric with respect to the plane $EAM_1$, and the planes $EHF$ and $ADB$ are parallel.

	Now let us show that points $H, D, B, F$ are coplanar and form a square $HDBF$. Indeed, all of its sides are have equal length as sides of a regular pentagon, and it has an axis of symmetry passing through the midpoints $M_1$ and $M_2$ of its opposite sides. Now if we glue the two halves of the polyhedron along this common square, the triangles $CDB$ and $ADB$ will be coplanar, since $\magl CM_2M_1 = \magl EM_1M_2$ and $\magl EM_1M_2 + \magl M_1M_2A = \pi$.

	Since $|AD| = |DC| = |CB| = |BA|$ as diagonals of a regular pentagon, $ADCB$ is a rhombus. By a similar argument, $EHGF$ is a rhombus as well. This completes the proof. \end{proof}

\bibliography{zba_ths}{}
\bibliographystyle{plain}

\end{document}